\documentclass[letterpaper, 10 pt, conference]{ieeeconf}   

\IEEEoverridecommandlockouts                              %

\usepackage{amsmath}
\usepackage{amssymb} 
\usepackage{comment}
\usepackage{cite}
\usepackage{todonotes}
\usepackage{graphicx}
\usepackage[ruled,linesnumbered]{algorithm2e}
\usepackage{tikz}
\usepackage{pgfplots}
\usetikzlibrary{shapes,arrows}
\definecolor{mst1}{rgb}{0.3266,    0.1840,    0.4290}
\definecolor{mst2}{rgb}{0.9665,    0.6340,   0.1526}
\definecolor{mst3}{rgb}{0.7629,    0.1073,    0.4411}
\usepackage{multirow} 
\usepackage{booktabs}
\usepackage{xcolor}

\newcommand{\funK}[2]{\boldsymbol K_{{#1},{#2}}}
\newcommand{\funVecK}[2]{\boldsymbol{k}_{{#1},{#2}}}
\newcommand{\funScalarK}[2]{k_{{#1},{#2}}}
\newcommand{\dataIndex}[2]{{#1}^{#2}}
\newcommand{\getElementOfVec}[2]{{[#1]}^{#2}}

\newcommand{\elementCounterB}{i}
\newcommand{\elementCounter}{i}
\newcommand{\dataCounter}{i}
\newcommand{\bfalpha}{\boldsymbol{\alpha}}
\newcommand{\1}{\boldsymbol{1}}
\newcommand{\0}{\boldsymbol{0}}
\newcommand{\M}{\boldsymbol{M}}
\newcommand{\I}{\boldsymbol{I}}
\newcommand{\W}{\boldsymbol{W}}
\newcommand{\Y}{\boldsymbol{Y}}
\newcommand{\Z}{\boldsymbol{Z}}
\newcommand{\X}{\boldsymbol{X}}
\newcommand{\x}{\boldsymbol{x}}
\newcommand{\y}{\boldsymbol{y}}
\newcommand{\w}{\boldsymbol{w}}
\newcommand{\meanGPmean}{\lambda}
\newcommand{\priorGPmean}{\gamma}
\newcommand{\meanGPcov}{\Lambda}
\newcommand{\priorGPcov}{\Gamma}
\newcommand{\mubf}{\boldsymbol{\mu}}
\newcommand{\Sigmabf}{\boldsymbol{\Sigma}}
\newcommand{\meanGPmeanVec}{\boldsymbol{\lambda}}
\newcommand{\priorGPmeanVec}{\boldsymbol{\gamma}}
\newcommand{\meanGPmeanDiag}{\bar{\boldsymbol{\lambda}}}
\newcommand{\priorGPmeanDiag}{\bar{\boldsymbol{\gamma}}}

\newcommand{\priorGroundTruth}{g} %
\newcommand{\meanGroundTruth}{h} %

\newtheorem{remark}{Remark}
\newtheorem{lemma}{Lemma}
\newtheorem{theorem}{Theorem}

\title{\LARGE \bf
Gaussian Processes with State-Dependent Noise for Stochastic Control
}

\author{Marcel Menner and Karl Berntorp%
\thanks{Marcel Menner and Karl Berntorp are with Mitsubishi Electric Research Laboratories (MERL), 201 Broadway, Cambridge, MA, 02139, USA (e-mail: menner@ieee.org; karl.o.berntorp@ieee.org).}%
}

\begin{document}

\maketitle
\thispagestyle{empty}
\pagestyle{empty}

\begin{abstract}
This paper considers a stochastic control framework, in which the residual model uncertainty of the dynamical system is learned using a Gaussian Process (GP). In the proposed formulation, the residual model uncertainty consists of a nonlinear function and state-dependent noise. The proposed formulation uses a \textit{posterior-GP} to approximate the residual model uncertainty and a \textit{prior-GP} to account for state-dependent noise. The two GPs are interdependent and are thus learned jointly using an iterative algorithm. Theoretical properties of the iterative algorithm are established. Advantages of the proposed state-dependent formulation include (i) faster convergence of the GP estimate to the unknown function as the GP learns which data samples are more trustworthy and (ii) an accurate estimate of state-dependent noise, which can, e.g., be useful for a controller or decision-maker to determine the uncertainty of an action. Simulation studies highlight these two advantages.
\end{abstract}

\section{Introduction}
Recent progress in computational resources, data science, robotics, and control has sparked an interest in considering uncertainty bounds within a controller or decision-maker~\cite{hewing2020learning}. %
Gaussian-process (GP) regression is a popular paradigm for estimating a function, because it provides not only a mean estimate of the desired quantity, but also a confidence bound~\cite{rasmussen2005gaussian}.
Using a probabilistic interpretation, a GP can accurately reflect the distribution of the data for a given and/or constant prior on the measurement noise. 
This paper proposes a GP formulation for learning an unknown function and its state-dependent noise distribution for stochastic control.

\subsection{Motivating Example}

GP regression uses the posterior distribution given data points to fit a function mapping from a value $\x$ to a noisy output $y$. 
Fig.~\ref{fig:motEx} illustrates a comparison between a classical GP and the proposed GP with state-dependent noise.
It shows that the classical solution fails to reduce its uncertainty bounds based on the proximity of the output measurements, $y$.
This behavior becomes apparent by inspecting the mean and covariance equations for a classical GP solution, in which the mean is a function of both $\x$ and the associated measurement $y$, i.e., $\boldsymbol \mu(\x)$ depends on the data $\x^i,y^i$, whereas the covariance $\boldsymbol \Sigma(\x)$ depends only on the input to the unknown function $\x^i$, and not on $y^i$. 
In contrast, the proposed GP formulation accurately reflects the noise in the data for both regions of higher noise and regions of lower noise.
\begin{figure}[t]
    \centering   
    \includegraphics[trim={1.1cm 0 1.2cm 0 } , clip,width=.9\columnwidth]{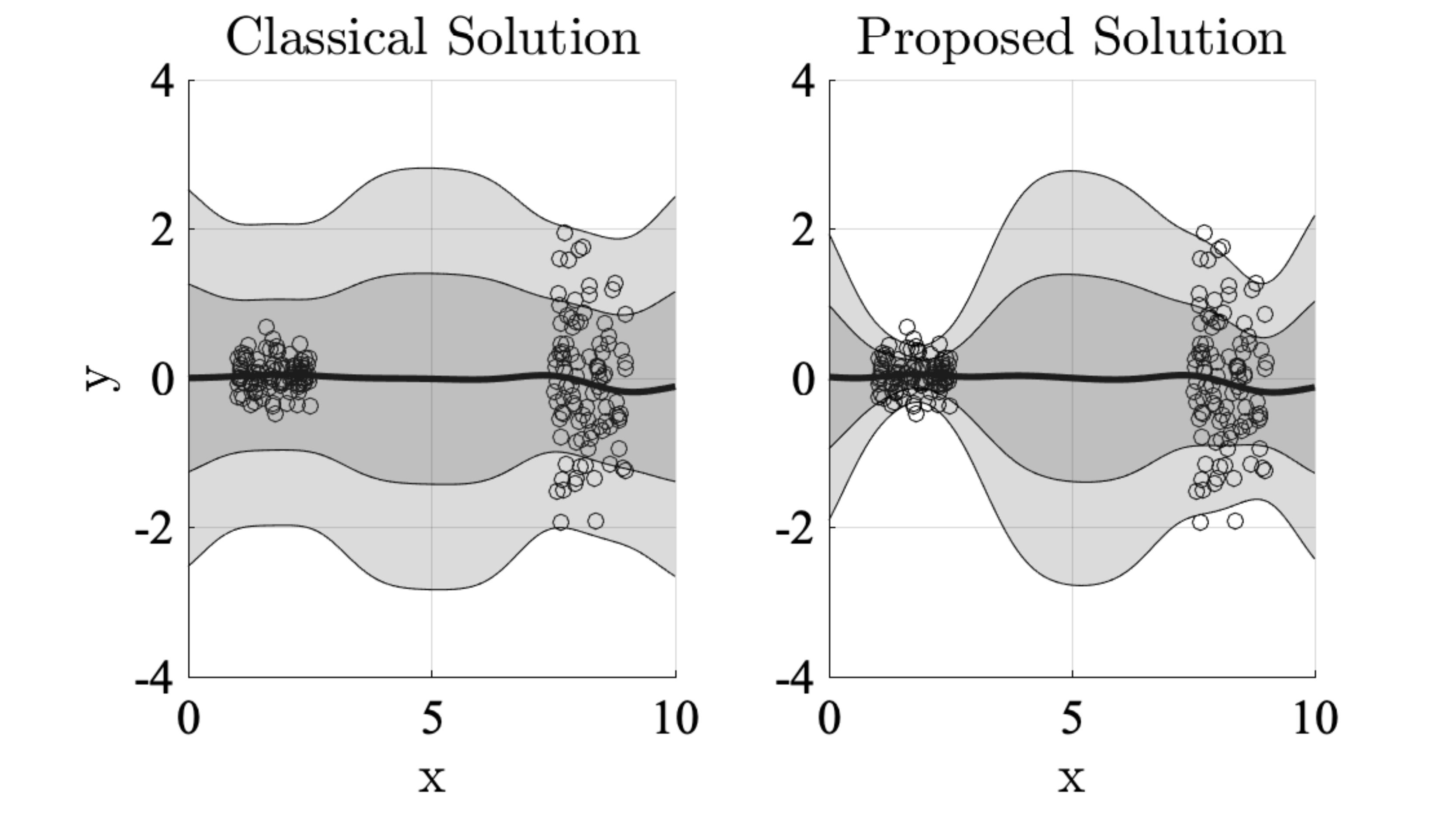}
    \caption{Motivating example for GP regression with state-dependent noise.
    The displayed scenario uses 100 samples between $x=[1,2.5]$ with lower noise and between %
    $x=[7.5,9]$ with higher noise. %
    Left: Classical solution without state-dependence.
    Right: State-dependent solution.
    The example that the GP on the left evenly reduces uncertainty based on the state, $x$.
    The GP on the right reduces uncertainty based on the state, $x$, but also based on the value, $y$.
    Hence, if the measurement values $y$ are ``closer together" (see left part of the two plots), the uncertainty is reduced appropriately.
    }
    \label{fig:motEx}
\end{figure}

\subsection{Contributions and Algorithmic Realization}
This paper considers a stochastic control framework, where the residual model uncertainty is learned using GP regression.
We propose a formulation, where the residual model uncertainty consists of a nonlinear function and state-dependent noise.
The proposed GP formulation learns both the nonlinear function and the noise distribution simultaneously.
It uses one GP to approximate the posterior distribution, which we refer to as the \textit{posterior-GP}.
Different from a classical GP, the proposed GP uses a data-based prior estimate to adjust the prior to match the noise distribution in the training data.
In this paper, this data-based prior estimate is provided by a second GP, which we refer to as \textit{prior-GP}.
The proposed posterior-GP uses the training data directly, and the prior-GP uses the variance of the training data.
However, the variance requires a mean estimate---provided by the posterior-GP.
On the other hand, the posterior-GP requires an estimate of the prior---provided by prior-GP.
We propose an iterative algorithm to address this interdependence and learn both GPs simultaneously. 
We present theoretical properties of the iterative algorithm such as boundedness and guaranteed convergence under some simplifying assumptions.
The main advantages of the proposed solution are (i) less data are needed to accurately estimate the unknown function due to the GP learning which data samples are more trustworthy and (ii) a decision-maker can accurately determine the uncertainty of an action.
Both advantages are illustrated using simulation examples.  

\subsection{Related Works}

\subsubsection*{GPs for stochastic control}
GPs are often used learning unknown parts of system dynamics~\cite{hewing2018cautious, hewing2020cautious,kabzan2019learning ,vaskov2022friction,mosharafian2021gaussian} or cost functions~\cite{levine2011nonlinear, levine2012continuous}.
In~\cite{hewing2018cautious, hewing2020cautious}, a GP is used to learn a residual dynamics model and its uncertainty bound is subsequently used within a chance-constrained MPC.
Further,~\cite{hewing2020cautious} presents approximations for efficient computation. %
In~\cite{kabzan2019learning}, a control approach based on learning GPs online is presented and applied to an autonomous race car. 
The work in~\cite{vaskov2022friction} uses a GP to model tire friction for improving the performance and robustness of a stochastic MPC.
In~\cite{mosharafian2021gaussian}, a GP-based stochastic MPC is used for cooperative adaptive cruise control, in which the uncertainty is used predictive behavioral models. 
In~\cite{levine2011nonlinear,levine2012continuous}, GPs are utilized to learn a cost function that can be used for control. 

\subsubsection*{Constraints for GP regression}
Recent work studies constrained GP regression~\cite{jidling2017linearly, maatouk2017gaussian,matschek2020constrained, berntorp2022online, swiler2020survey, da2012gaussian, da2020gaussian}, which can, e.g., be used to enforce physics-based information.
E.g., the approach in~\cite{jidling2017linearly} uses a linear operators to enforce equality constraints,~\cite{maatouk2017gaussian} uses virtual measurements to enforce inequality constraints,~\cite{matschek2020constrained} uses tractability constraint to learn a reference trajectory from data for model predictive control, and~\cite{berntorp2022online} uses a basis-function approach to enforce constraints on the GP model.

\subsubsection*{GPs with input-dependent noise} 
Input-dependent noise is considered in~\cite{goldberg1997regression, kersting2007most, boukouvalas2014optimal, miyagusuku2015gaussian, zhang2020improved}, e.g., for robotic perception and localization. 
In~\cite{goldberg1997regression}, the posterior distribution of the noise distribution is sampled using Markov chain Monte Carlo methods.
In contrast, we use an algorithm that does not rely on sampling and iteratively refines both the posterior-GP and the prior-GP.
The work in~\cite{kersting2007most} also considers an iterative estimation algorithm, which is not guaranteed to converge and may lead to oscillations. 
Compared to~\cite{kersting2007most}, we state theoretical properties of the proposed iterative algorithm such as boundedness and convergence under some mild assumptions.
Further,~\cite{boukouvalas2014optimal} considers input-dependent noise for optimal design, and~\cite{miyagusuku2015gaussian} presents a robot localization task using GPs with input-dependent noise.
Different from~\cite{goldberg1997regression, kersting2007most, boukouvalas2014optimal, miyagusuku2015gaussian, zhang2020improved}, we present GPs with state-dependent noise with application to stochastic control and present theoretical properties of the estimation algorithm.

\subsection{Preliminaries and Notation}
We use $\getElementOfVec{\x}{\elementCounter}$ to indicate the $\elementCounter$-th element of a vector $\x$ and $\getElementOfVec{\X}{\elementCounter\elementCounter}$ to indicate the element in the $\elementCounter$-th row and $\elementCounter$-th column of a matrix $\X$.
We use the Hadamard product (element-wise product) $\circ$, i.e., for $\x, \w \in \mathbf{R}^{n_x}$, the $\elementCounter$-th element of $\boldsymbol z = \x \circ \w \in \mathbf{R}^{n_x}$ is 
$\getElementOfVec{\boldsymbol z}{\elementCounter} = \getElementOfVec{\x}{\elementCounter} \getElementOfVec{\w}{\elementCounter}$.
To ease exposition, we define $\x^{\circ 2}:= \x \circ \x$.
We use $\I$ for the identity matrix and $\1$ and $\0$ for the one and zero vector of appropriate dimension. For a vector  $\x$, $\x \sim \mathcal{N}(\mubf,\Sigmabf)$  
indicates that  $\x\in \mathbf{R}^{n_x}$ is Gaussian distributed with mean $\boldsymbol \mu$ and covariance $\boldsymbol\Sigma$.

Throughout, we use a common choice of kernel $k(\x,\w)$ with the radial basis function,
\begin{align}
\label{eq:kernel}
    k(\x,\w)
    =
    \exp{
    \left(
    - l (\x-\w)^T(\x-\w)
    \right)
    },
\end{align}
where $l$ is the length scale defining the spread of the basis function.
Note that $k(\x,\w) = 1$ if $\x=\w$ and $k(\x,\w) \rightarrow 0$ for $\x$ far away from $\w$.
For brevity, let $\funScalarK{\x}{\w}:=k(\x,\w)$.
Next, we define
\begin{align*}
    \funVecK{\x}{\W}
    =
    \begin{bmatrix}
    \funScalarK{\x}{\dataIndex{\w}{1}}  
    &
    \funScalarK{\x}{\dataIndex{\w}{2}}
    &
    \hdots
    &
    \funScalarK{\x}{\dataIndex{\w}{M}}
    \end{bmatrix}
\end{align*}
and
\begin{align*}
    \funK{\X}{\W}
    =
    \begin{bmatrix}
    \funScalarK{\dataIndex{\x}{1}}{\dataIndex{\w}{1}}  
    & 
    \funScalarK{\dataIndex{\x}{1}}{\dataIndex{\w}{2}}  
    & \hdots & 
    \funScalarK{\dataIndex{\x}{1}}{\dataIndex{\w}{M}}  
    \\
    \funScalarK{\dataIndex{\x}{2}}{\dataIndex{\w}{1}}  
    & 
    \funScalarK{\dataIndex{\x}{2}}{\dataIndex{\w}{2}}  
    & \hdots & 
    \funScalarK{\dataIndex{\x}{2}}{\dataIndex{\w}{M}}  
    \\
    & \vdots
    \\
    \funScalarK{\dataIndex{\x}{N}}{\dataIndex{\w}{1}}  
    & 
    \funScalarK{\dataIndex{\x}{N}}{\dataIndex{\w}{2}}  
    & \hdots & 
    \funScalarK{\dataIndex{\x}{N}}{\dataIndex{\w}{M}}  
    \end{bmatrix}
\end{align*}
with the collection of vectors
\begin{align*}
    \X
    &=
    \begin{bmatrix}
    \dataIndex{\x}{1} & \dataIndex{\x}{2} \hdots & \dataIndex{\x}{N} 
    \end{bmatrix} \in \mathbf{R}^{n_x\times N}
    \\
    \W
    &=
    \begin{bmatrix}
    \dataIndex{\w}{1} & \dataIndex{\w}{2} \hdots & \dataIndex{\w}{M} 
    \end{bmatrix} \in \mathbf{R}^{n_x\times M}.
\end{align*} 

For a function $\meanGPmean(\x)$ with $\meanGPmean:\mathbf{R}^{n_x}\rightarrow \mathbf{R}$, we use bold symbol $\meanGPmeanVec(\X)$ with $\meanGPmeanVec : \mathbf{R}^{n_x\times N}\rightarrow \mathbf{R}^N$ as shorthand notation to indicate that the $\elementCounter$-th element of 
$\getElementOfVec{\meanGPmeanVec(\X)}{\elementCounter}
= \meanGPmean(\dataIndex{\x}{\elementCounter})$. Similarly, we use the diagonal matrix $\meanGPmeanDiag(\X)$ with $\meanGPmeanDiag:\mathbf{R}^{n_x\times N}\rightarrow \mathbf{R}^{N\times N}$, whose diagonal elements $\getElementOfVec{\meanGPmeanDiag(\X)}{\elementCounter\elementCounter}=\meanGPmean(\dataIndex{\x}{\elementCounter})$. Hence, $\meanGPmeanDiag(\X)={\rm diag}(\meanGPmeanVec(\X))$.

\section{Problem Statement} 
\subsection{Stochastic Control Formulation}

We consider an uncertain dynamical system with the measurable state $\x_k\in\mathbf{R}^{n_x}$ in discrete time,
\begin{align}
\label{eq:dyn_sys}
    \x_{k+1} = \boldsymbol f(\x_k,\boldsymbol u_k) + \boldsymbol B \left(\boldsymbol h(\x_k,\boldsymbol u_k) + \boldsymbol \epsilon_k\right),
\end{align}
where $\boldsymbol u_k\in\mathbf{R}^{n_u}$ is the controllable input. 
The dynamical system model in~\eqref{eq:dyn_sys} is composed of a known nominal function $\boldsymbol f$, the known matrix $\boldsymbol{B}$ with full-column rank, and an initially unknown function $\boldsymbol h:\mathbf{R}^{n_x}\times \mathbf{R}^{n_u}\rightarrow\mathbf{R}^{n_h}$, which is to be learned using data.
Finally, $\boldsymbol \epsilon_k\in\mathbf{R}^{n_\epsilon}$ denotes noise or uncertainty arising from imperfect sensors, perception, model uncertainty, or the like.
We assume that the elements in $\boldsymbol \epsilon_k$ are uncorrelated, i.e., $\boldsymbol \epsilon_k$ is a realization of a distribution with a diagonal covariance matrix, however, the diagonal entries are allowed to be state-dependent,
\begin{align}
\label{eq:disturbance_definition}
    \boldsymbol \epsilon_k \sim \mathcal{N}\left(\0,
    \boldsymbol\Sigma(\x_k,\boldsymbol u_k)\right)
\end{align}
with the diagonal matrix $\boldsymbol\Sigma(\x_k,\boldsymbol u_k)$. 

The dynamical system is subject to state and input constraints, which are enforced as chance constraints,
\begin{subequations}
\label{eq:chance_constraints_problem}
    \begin{align}
        &{\rm Pr}(\x_k\in \mathrm{X})\geq p_x
        \\
        &{\rm Pr}(\boldsymbol u_k\in \mathrm{U})\geq p_u,
    \end{align}
\end{subequations}
where $\mathrm{X}$ and $\mathrm{U}$ are the state and input constraint sets that are to be satisfied with $p_x$ and $p_u$ probability, respectively.

The goal is thus to control the dynamical system in~\eqref{eq:dyn_sys} in the presence of the noise/disturbances in~\eqref{eq:disturbance_definition} such that the chance constraints in~\eqref{eq:chance_constraints_problem} are not violated. 
In the following, we drop the dependence of $\boldsymbol h$ and $\boldsymbol\Sigma$ on $\boldsymbol u_k$ to ease exposition.
In order to learn the unknown function $\boldsymbol h$ and unknown covariance matrix $\boldsymbol\Sigma$, we use~\eqref{eq:dyn_sys} to form
\begin{align*}
    \y = 
    \boldsymbol B^{+}\left(
    \x_{k+1} - \boldsymbol f(\x_k,\boldsymbol u_k)
    \right)
    =
    \boldsymbol h(\x_k) + \boldsymbol\epsilon
\end{align*}
with the Moore-Penrose pseudo inverse $\boldsymbol B^{+}$.

\subsection{Gaussian-Process Regression}

As the covariance in~\eqref{eq:disturbance_definition} is diagonal, we present the estimation problem for the one-dimensional case $n_h=1$ for simplicity.
However, $n_h\geq 2$ follows immediately by using $n_h$ GPs, i.e., one GP for each of the $n_h$ dimensions. 
Hence, we use 
\begin{subequations}
\begin{align}
\label{eq:unknwon_func}
    y^i =  \meanGroundTruth(\x^i) + \epsilon^i,
\end{align}
where $i$ indicates the $i^{\rm th}$ data point and
\begin{align}
\label{eq:unknwon_noise}
    \epsilon^i \sim \mathcal{N}\left(0, \priorGroundTruth (\x^i) \right),
\end{align}
\end{subequations}
with state-dependent variance $\priorGroundTruth (\x^i)>0$ of the noise $\epsilon^i$.

This paper addresses  how to simultaneously learn a surrogate function of $h$ with state-dependent variance $g$ using GPs with the training data
\begin{subequations}
\label{eq:training_data}
\begin{align}
    \mathcal{D}_\meanGPmean 
    & = \{ \dataIndex{\x}{\dataCounter},\ \dataIndex{y}{\dataCounter}\}_{\dataCounter=1}^N
    \\
    \X
    &=
    \begin{bmatrix}
    \dataIndex{\x}{1} & \dataIndex{\x}{2} \hdots & \dataIndex{\x}{N} 
    \end{bmatrix} \in \mathbf{R}^{n_x\times N}
    \\
    \label{eq:data_Y}
    \Y 
    &= \begin{bmatrix}
        \dataIndex{y}{1} & \dataIndex{y}{2} & \hdots & \dataIndex{y}{N}
    \end{bmatrix}^T \in \mathbf{R}^{N},
\end{align}
\end{subequations}
where $N$ is the number of data samples.

\section{Mathematical Formulation of Gaussian Processes with State-Dependent Noise}
This section presents the proposed learning algorithm, in which one GP is used to approximate the unknown function $h$ in~\eqref{eq:unknwon_func}, which we refer to as \textit{posterior-GP}.
In order to reflect state-dependent noise, the posterior-GP uses a prior estimate of the noise distribution.
In this paper, this prior estimate is provided by a second GP, which adjusts the prior in order to reflect the noise distribution and is thus referred to as \textit{prior-GP}.
Both the posterior-GP and the prior-GP are determined using the data in~\eqref{eq:training_data}, where the posterior-GP uses the data in~\eqref{eq:data_Y} directly, and the prior-GP uses the variance of the data in~\eqref{eq:data_Y}.  
However, computing the variance requires a mean estimate---provided by the posterior-GP. In turn, the mean estimate depends on the prior---provided by the prior-GP. Hence, the posterior-GP and the prior-GP are interdependent, which we address next. 

\subsection{Posterior-GP}
In order to derive the posterior distribution of the unknown function in~\eqref{eq:unknwon_func} given the measurement data, we use
\begin{align*}
    \begin{bmatrix}
    \Y  
    \\ 
    y
    \end{bmatrix}
    \sim
    \mathcal{N}
    \left(
    \0,
    \begin{bmatrix}
    \funK{\X}{\X} \!+\! \sigma_0^2 \I \!+\! \priorGPmeanDiag(\X)
    & 
    \funVecK{\X}{\x}
    \\ 
    \funVecK{\x}{\X} &\funScalarK{\x}{\x} \!+\! \sigma_0^2 \!+\! \priorGPmean(\x)
    \end{bmatrix}
    \right)
\end{align*} 
where $\priorGPmean(\x)$ is an initially unknown function, which assigns a prior to the measurement data.
Then, the conditional distribution $p(y|\Y,\X,\x)$ is a normal distribution with
\begin{align*}
y\sim \mathcal{N}\left( \meanGPmean(\x),  \meanGPcov(\x)\right),   
\end{align*}
where the mean and the posterior variance are given by
\begin{subequations}
\label{eq:meanGP}
\begin{align}
\label{eq:meanGP_mean}
    \meanGPmean(\x) &= \funVecK{\x}{\X}\left( \funK{\X}{\X} + \sigma^2_0 \I + \priorGPmeanDiag(\X) \right)^{-1} \Y 
    \\
    \meanGPcov(\x) &= 
    \funScalarK{\x}{\x} \!+\! \sigma_0^2 \!+\! { \priorGPmean}(\x)
    \nonumber
    \\
    &\quad
    - \funVecK{\x}{\X}\left( \funK{\X}{\X} \!+\! \sigma_0^2 \I \!+\! \priorGPmeanDiag(\X) \right)^{-1} \funVecK{\X}{\x}.\label{eq:meanGP_cov}
\end{align}
\end{subequations}

\begin{remark}[Connection to classical GPs]
In classical GP regression, the prior is assumed to be known and $\boldsymbol{\bar \priorGPmean}(\boldsymbol x)\!=\!\mathbf{0}$. We will use this as baseline comparison in Section~\ref{sec:results}.  
\end{remark}

\subsection{Prior-GP}
To infer an accurate uncertainty quantification or noise estimate from data, this paper proposes to utilize a prior-GP.
The prior-GP is trained using the measurements, $\dataIndex{y}{\dataCounter}$, and the mean estimate of the posterior-GP, $\meanGPmean(\x)$ in~\eqref{eq:meanGP_mean}.
Hence, the prior-GP uses training data computed as the second moment of the measurements, $\dataIndex{y}{\dataCounter}$, i.e.,
\begin{subequations}
\label{eq:data_train_var}
\begin{align}
    \mathcal{D}_\priorGPmean = \{\dataIndex{\x}{\dataCounter},\ \dataIndex{z}{\dataCounter}\}_{\dataCounter=1}^N,
\end{align}
with
\begin{align} 
    \dataIndex{z}{\dataCounter}
    &=
    \left(\dataIndex{y}{\dataCounter} - \meanGPmean(\dataIndex{\x}{\dataCounter})\right)^2 
    -
    \sigma_0^2,
    \\
    \Z 
    &= \begin{bmatrix}
        \dataIndex{z}{1} & \dataIndex{z}{2} & \hdots & \dataIndex{z}{N}
    \end{bmatrix}^T  \in \mathbf{R}^{N}.
\end{align}
\end{subequations}
Then, the prior-GP uses the data in~\eqref{eq:data_train_var} and 
\begin{align*}
    \begin{bmatrix}
    \Z %
    \\ 
    z
    \end{bmatrix}
    \sim
    \mathcal{N}
    \left(
    \boldsymbol 0,
    \begin{bmatrix}
    \funK{\X}{\X} + \sigma_0^2 \boldsymbol I 
    & 
    \funVecK{\X}{\x}
    \\ 
    \funVecK{\x}{\X} & \funScalarK{\x}{\x} + \sigma_0^2 
    \end{bmatrix}
    \right).
\end{align*}
Using the conditional distribution $p(z|\Z,\X,\x)$, the prior-GP is given by
\begin{align*}
z \sim \mathcal{N}\left( \priorGPmean(\x),  \priorGPcov(\x)\right),   
\end{align*}
with
\begin{subequations}
\label{eq:prior_gp}
\begin{align}
\label{eq:prior_gp_mean}
    \priorGPmean(\x) &= \funVecK{\x}{\X} \left( \funK{\X}{\X} + \sigma^2_0 \boldsymbol I \right)^{-1} \Z,
    \\ 
    \priorGPcov(\x) &= \funScalarK{\x}{\x} \!+\! \sigma^2_0 - \funVecK{\x}{\X}\left( \funK{\X}{\X} \!+\! \sigma^2_0 \boldsymbol I \right)^{-1} \funVecK{\X}{\x}.\label{eq:prior_gp_cov}
\end{align}
\end{subequations}
 
On the one hand, the training data for the prior-GP in~\eqref{eq:data_train_var} require a mean estimate of the unknown function $h$.
On the other hand, the posterior-GP requires a prior estimate of the noise distribution.
Hence, the posterior-GP in~\eqref{eq:meanGP} and the prior-GP in~\eqref{eq:prior_gp} need to be estimated jointly, which we address using an iterative algorithm in Section~\ref{sec:iterative_algorithm}.

\begin{remark}[Exploitation-exploration]
    GP models are often used as a surrogate function for black-box optimization, e.g., in Bayesian optimization. 
    For this application, $\meanGPmean(\x)$ in~\eqref{eq:meanGP_mean} quantifies the mean estimate and $\priorGPcov(\x)$ in~\eqref{eq:prior_gp_cov} quantifies the uncertainty, which can be used for exploration, whereas $\meanGPcov(\x)$ in~\eqref{eq:meanGP_cov} quantifies the noise in the data.
\end{remark}

\begin{remark}[Parametric model]
An alternative to using a prior-GP is to use a parametric model, e.g., using basis functions. In this case,
\begin{align*}
    \priorGPmean(\boldsymbol x)
    =
    \boldsymbol \phi(\boldsymbol x) \boldsymbol \theta,
\end{align*}
where $\boldsymbol \phi : \mathbf{R}^{n_x}\rightarrow \mathbf{R}^{n_\phi}$ denotes the vector of $n_\phi$ basis functions and $\boldsymbol \theta$ is the vector of parameters to be learned.
The parameter vector, $\boldsymbol \theta$, can be learned using regression,
\begin{align*}
    \boldsymbol \theta &= \arg\min_{\boldsymbol {\tilde \theta}} \sum_{\dataCounter=1}^N \left(\dataIndex{z}{\dataCounter} - 
    \boldsymbol \phi\left(\dataIndex{\x}{\dataCounter} \right)\boldsymbol {\tilde \theta}\right)^2.
\end{align*}
Note that there is a connection between GPs and basis-function expansions through a Hilbert-space interpretation~\cite{Solin2020}.
\end{remark}

\subsection{Iterative algorithm}
\label{sec:iterative_algorithm}
We jointly estimate the unknown function in~\eqref{eq:unknwon_func} and the noise distribution of its data samples in~\eqref{eq:unknwon_noise} by iteratively refining the posterior estimate and the prior estimate.
Algorithm~\ref{alg:alg} summarizes the procedure.
At iteration~$j$, on Line~\ref{alg:update_mean}, the iterative algorithm first updates the posterior-GP using the mean estimate of the prior-GP at the previous iteration with
\begin{subequations}
\label{eq:iterations}
\begin{align}
\label{eq:iterations_mean}
    \meanGPmeanVec_{j+1}(\X) &= \funK{\X}{\X} \left(\funK{X}{X}+\sigma_0^2 \I + \priorGPmeanDiag_{j}(\X)\right)^{-1} \Y,
\end{align}
Then, on Line~\ref{alg:update_prior_data} it uses the updated posterior-GP estimate in~\eqref{eq:iterations_mean} to update the training data for the prior-GP,
\begin{align}  
\label{eq:iterations_data}
    \Z_{j+1} & = \left(\left(\Y\!-\!\meanGPmeanVec_{j+1}(\X)\right)^{\circ 2} \!-\! \sigma_0^2 \mathbf{1}\right),
\end{align}
Next, it uses the training data in~\eqref{eq:iterations_data} to update the prior GP,
\begin{align} 
    \priorGPmeanVec_{j+1}(\X) &= \funK{\X}{\X} \left(\funK{\X}{\X} \!+\!\sigma_0^2 \I \right)^{-1} 
    \Z_{j+1},
\end{align}
see Line~\ref{alg:update_prior} in Algorithm~\ref{alg:alg}.
\end{subequations}
Finally, Algorithm~\ref{alg:alg} is stopped if a stopping criterion is met. 
We use a stopping criterion based on the basis-function interpretation of GPs, i.e.,
\begin{align*}
    \bfalpha_{j+1} = (\funK{\X}{\X} + \sigma_0^2 \I + \priorGPmeanDiag_{j+1}(\X))^{-1}\Y,
\end{align*}
where $\bfalpha$ can be thought of as weights for the basis functions, i.e., $\meanGPmean(\x)=\funVecK{\x}{\X}\bfalpha$.
This choice is motivated by the impact that the prior-GP has on the posterior-GP.
In the following, we show that this iterative procedure converges to its optimum under some assumptions.

\begin{algorithm}
\label{alg:alg}
    \SetKwInOut{Input}{Input}
    \SetKwInOut{Output}{Output}
     \SetKwRepeat{Do}{do}{while}
   Initialize $\alpha_0=\mathbf{0}$, $\Z_0=\mathbf{0}$, $\priorGPmeanVec_{0}(\X) = \0$, $j=0$ \;   \label{alg:init}
   \Do{$\|\bfalpha_{j+1} - \bfalpha_{j}\|_2\leq \delta $}{
   \label{alg:iter}
   {\color{gray} \texttt{\%\% update posterior-GP estimate}\; }
   $\meanGPmeanVec_{j+1}(\X) \!=\! \funK{\X}{\X} \left(\funK{\X}{\X}
   \!+\! \sigma_0^2 \I \!+\! \priorGPmeanDiag_{j}(\X)\right)^{-1}\hspace{-0.05cm} \Y$\; 
        \label{alg:update_mean}
   {\color{gray} \texttt{\%\% update training data for prior-GP using posterior-GP}\; }
   $\Z_{j+1} = \left(\left(\Y\!-\!\meanGPmeanVec_{j+1}(\X)\right)^{\circ 2} \!-\! \sigma_0^2 \1 \right)$ \;
        \label{alg:update_prior_data}
   {\color{gray} \texttt{\%\% update prior-GP}\; }
   $\priorGPmeanVec_{j+1}(\X) = \funK{\X}{\X} \left(\funK{\X}{\X} \!+\!\sigma_0^2 \I \right)^{-1} 
    \Z_{j+1}$ \; 
        \label{alg:update_prior}
    {\color{gray} \texttt{\%\% check stopping criterion} \;}
    $\bfalpha_{j+1} = (\funK{\X}{\X} + \sigma_0^2 \I + \priorGPmeanDiag_{j+1}(\X))^{-1} \Y$ \;
    $j \leftarrow j+1 $ \;
        \label{alg:stopCrit}
   }  
    \caption{Iterative algorithm for GP regression with state-dependent noise}
\end{algorithm}

\section{Theoretical Properties of Algorithm}   
In order to analyze the convergence characteristics of the algorithm, we make a simplifying assumption, i.e., we study the algorithm in the presence of a small length scale $l$ in~\eqref{eq:kernel}.
Using a small length scale $l\rightarrow 0$, we can simplify $\funK{\X}{\X}\rightarrow \I$, which allows for formulating analytical expressions of the matrix inverses in~\eqref{eq:iterations}. 
Note that the assumption of a small length scale is only needed to establish a series of (nontight) bounds for the optimizer in Algorithm~\ref{alg:alg}.
Moreover, the assumption is utilized to show how to choose~$\sigma_0^2$ appropriately.
Therefore, in practice, this assumption is not needed.

The main steps of the proof are outlined in Lemma~\ref{lma:joint_convergence}, Lemma~\ref{lma:bounded}, and Theorem~\ref{thm:contracting}.
In particular, 
Lemma~\ref{lma:joint_convergence} shows that if the prior-GP converges, then the posterior-GP converges as well;
Lemma~\ref{lma:bounded} shows that the estimated prior-GP remains bounded for all iterations;
and Theorem~\ref{thm:contracting} proves that the prior GP is contracting, i.e., the prior GP converges to its optimum.

\begin{lemma}
\label{lma:joint_convergence}
If prior-GP converges, then the posterior-GP converges as well, i.e., if $\priorGPmeanVec_j(\X)\!\rightarrow \!\priorGPmeanVec_\star(\X)$, then $\meanGPmeanVec_j(\X)\!\rightarrow\! \meanGPmeanVec_\star(\X)$.
\end{lemma}
\begin{proof}
    This follows immediately from the iterations~\eqref{eq:iterations} with $\meanGPmeanVec_{\star}(\X) = \funK{\X}{\X} \left(\funK{\X}{\X} +\sigma_0^2 \I + \priorGPmeanDiag_{\star}(\X)\right)^{-1} \Y$. 
\end{proof}
\begin{lemma}
\label{lma:bounded}
Let the prior-GP be initialized as in Line~\ref{alg:init} in Algorithm~\ref{alg:alg}.
For a small length scale $l\rightarrow 0$, the prior estimate  $\priorGPmeanVec_j(\X)$ is bounded below and above with
\begin{align*}
    - \frac{\sigma_0^2}{1+\sigma_0^2}\1 & 
    \leq 
    \priorGPmeanVec_{j}(\X) 
    \leq 
    \frac{1}{1 + \sigma_0^2} (  \Y^{\circ 2} - \sigma_0^2 \1 )
\end{align*}
 for all iterations~$j$.
\end{lemma}
\begin{proof}[Sketch]
    The proof detailed in the Appendix uses induction arguments in order to show that the iterations~$j$ of Algorithm~\ref{alg:alg} remain bounded.
\end{proof}

\begin{theorem}
\label{thm:contracting}
Let $\sigma_0^2$ be chosen such that $\sigma_0^2 \1\geq \Y^{\circ 2}$ and let $\priorGPmeanVec_\star(\X)$ denote the optimum of the prior-GP.
Then, using the assumption of a small length scale, we can show that Algorithm~\ref{alg:alg} converges to the optimum, $\priorGPmeanVec_{j}(\X) \rightarrow \priorGPmeanVec_\star(\X)$ for $j\rightarrow\infty$.
\end{theorem} 
\begin{proof}[Sketch]
The proof detailed in the Appendix uses the bounds established in Lemma~\ref{lma:bounded} in order to show that the prior-GP is contracting with
\begin{align*}
    \priorGPmeanVec_{j+1}(\X) - \priorGPmeanVec_\star(\X) = \boldsymbol A_j (\priorGPmeanVec_j(\X) - \priorGPmeanVec_\star(\X)),
\end{align*}
i.e., the eigenvalues of $\boldsymbol A_j$ are inside the unit circle for every iteration~$j$ of Algorithm~\ref{alg:alg}.
\end{proof}

\begin{remark}
The design choice and condition for convergence $\sigma_0^2 \1\geq \Y^{\circ 2}$ makes sense intuitively, because it means that the prior $\sigma_0^2$ should be chosen to cover the square root of the data points~$\dataIndex{y}{\dataCounter}$. In other words, the prior should have at least the support of the posterior, akin to why proposal densities work for sequential Monte-Carlo methods \cite{Doucet2009}. 
\end{remark}

\section{Simulation Results}
\label{sec:results}
In this section, we validate our proposed method using two numerical examples. The first is a non-physical system meant to illustrate the efficacy of the proposed method. The second example shows how to leverage the state-dependent variance estimates for stochastic control.

\subsection{Illustrative Estimation Example}
In order to study the proposed method, we generate data according to the ground-truth function
\begin{subequations} 
\label{eq:ground_truth_results}
\begin{align}
    \dataIndex{y}{\dataCounter} &= 0.5(\cos{\dataIndex{x}{\dataCounter}} + 1) + \dataIndex{\epsilon}{\dataCounter}
    \\
    \dataIndex{\epsilon}{\dataCounter} &\sim \mathcal{N}\left(0, 0.19 \dataIndex{x}{\dataCounter} + 0.1\right),
\end{align}
\end{subequations}
where the data points $\dataIndex{x}{\dataCounter}$ are sampled from a uniform distribution $\dataIndex{x}{\dataCounter}\! \sim\! \mathcal{U}(0,10)$. 
Next, we study the number of data samples  required to accurately approximate both the mean and the noise distribution in~\eqref{eq:ground_truth_results}.

\subsubsection{Qualitative results}
Fig.~\ref{fig:increase_N} illustrates a comparison of a classical implementation of a GP and the proposed solution for different numbers of collected data points $N$.
First, consider $N\!=\!1$.
As the one collected data point is far from zero (around the $2\sigma_0$ bound), the proposed solution does not trust the data point as much as the classical solution. 
Hence, in such an instance, the proposed solution is more conservative.
Next, consider $N\!=\!5$.
Here, it can be seen that the ``closeness" of the four measurements between $x\in [0,5]$ gives the proposed Gaussian process regression algorithm confidence in its estimate. Consequently, the uncertainty is decreased more than for the classical GP.
For $N\!=\!20$, the proposed algorithm increases the uncertainty between $x\in [5,10]$ as measurements are more spread out. 
Here, it can be seen that the classical GP underestimates the noise for $x\in [5,10]$ and overestimates the noise for $x\in [0,5]$. In contrast, the proposed GP correctly predicts the trend of the uncertainty in the data of the unknown function.
For $N\!=\!100$, the proposed solution accurately estimates both the mean function and its noise distribution, whereas the classical GP solution can only approximate the mean.

\begin{figure}[t]
    \centering   
    \includegraphics[trim={0.8cm 2.2cm 1.5cm 1.2cm}, clip, width=1\columnwidth]{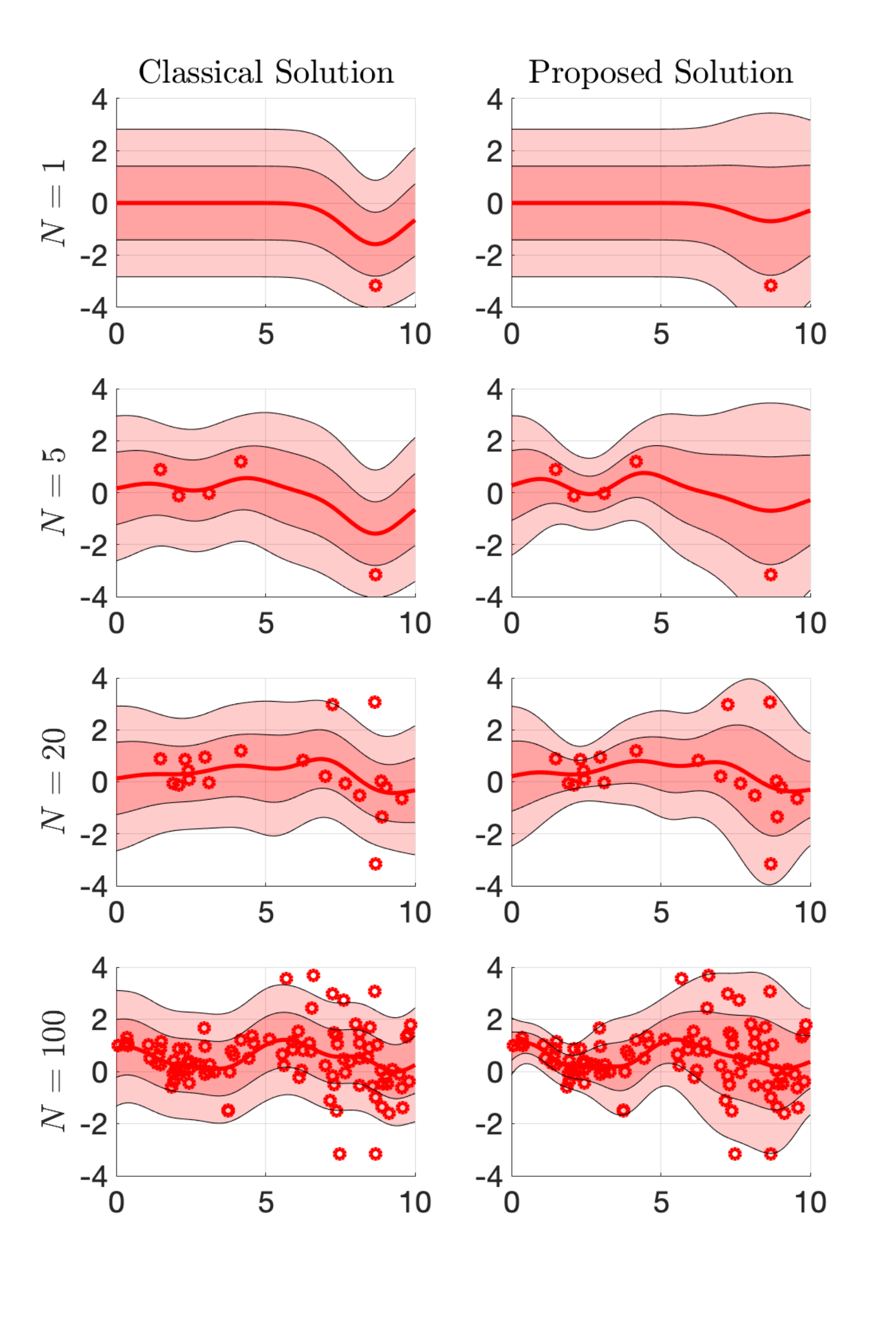}
    \caption{Comparison of classical GP regression and the proposed GP regression with state-dependent noise. 
    Left column: Classical GP regression.
    Right column: Proposed GP regression.
    The plot illustrates the mean in red as well as the $1\sigma$ bound and $2\sigma$  bound in different shades of red.
    The training data points are illustrated in and are the same for both classical GP regression and the proposed method.}
    \label{fig:increase_N}
\end{figure}

\subsubsection{Quantitative results}
Fig.~\ref{fig:conv} shows the convergence rate of the GP to the ground truth function, i.e., $\|\meanGPmean(x)- \meanGroundTruth(x)\|_2$ and $\|\priorGPmean(x)-\priorGroundTruth(x)\|_2$.
It shows that the proposed GP with state-dependent noise is able to learn both the mean and the variance of the ground-truth function.
First, the convergence rate of the mean estimate in the proposed method is faster than the standard GP implementation with a constant prior. 
The reason for the faster convergence rate is that the proposed GP implicitly learns that it can trust the measurements more if they are closer together, and also that it cannot trust the measurements as much if they are further apart.
As $N\rightarrow \infty$, the mean estimates of both the classical GP and the proposed GP formulations converge to the ground truth. 
However, the variance estimates of the constant-prior implementation will inaccurately reflect the true uncertainty, which can be detrimental when used in combination with a feedback controller.

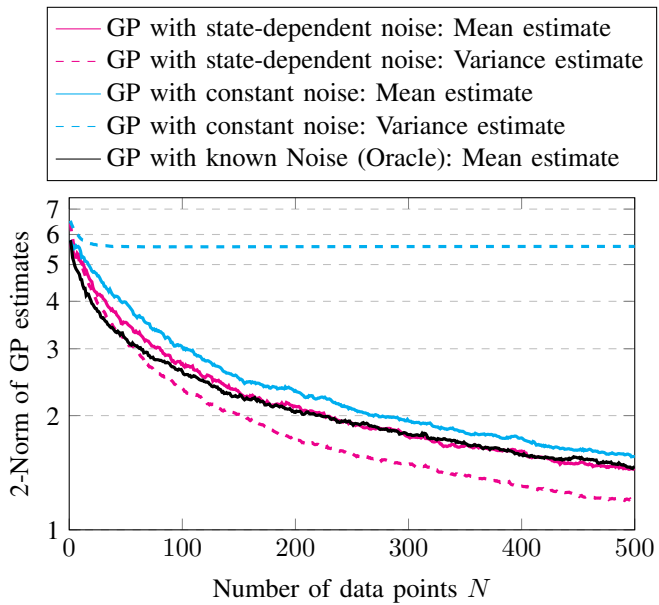
\begin{figure}[t]
\centering   
\begin{tikzpicture}[every plot/.append style={ very thick}]
\begin{axis}[
    ymode=log,
    xmin=0, xmax=500,
    ymin=1, ymax=7.5, 
   ylabel={2-Norm of GP estimates}, 
   ylabel near ticks,
   xlabel={Number of data points $N$},  
   ytick={1,2,3,4,5,6,7},  
   yticklabels={1,2,3,4,5,6,7},  
    ymajorgrids=true,
    legend style={at={(0.5,1.05)},anchor=south},  
    legend columns=1,
    grid style=dashed,
    width=9.1cm, 
    height=6cm,
    legend cell align={left}
	] 
]	  
\addplot[color=magenta] table[x index=0,y index=1] {data/convProposedGP.txt};
\addplot[color=magenta,dashed] table[x index=0,y index=1] {data/convProposedGPvar.txt};
\addplot[color=cyan] table[x index=0,y index=1] {data/convVanillaGP.txt};
\addplot[color=cyan,dashed] table[x index=0,y index=1] {data/convVanillaGPvar.txt};  
\addplot[color=black] table[x index=0,y index=1] {data/convOracleGP.txt};

\legend{GP with state-dependent noise: Mean estimate, 
GP with state-dependent noise: Variance estimate, 
GP with constant noise: Mean estimate,
GP with constant noise: Variance estimate, 
GP with known Noise (Oracle): Mean estimate};     
\end{axis}
\end{tikzpicture}  
\caption{Statistics of GP estimate accuracy. The plot illustrates the median of the 2-norm of 500 Monte-Carlo trials.
The plot shows how quickly the mean estimate can decrease if the noise distribution were known (in black), which provides a bound on the convergence rate.
Further, the plot shows two GP models: a GP with constant noise (in cyan) and the proposed GP with state-dependent noise (in magenta).
}
\label{fig:conv}
\end{figure}

\subsection{Application to Stochastic Control}
Next, we apply the proposed GP regression model to an illustrative stochastic control example.
We employ the solver in~\cite{menner2022learning}. 
However, any stochastic control method, see e.g., \cite{hewing2020cautious}, can be utilized as long the control formulation considers uncertainty bounds.  
We consider the dynamical system
\begin{align*}
\begin{bmatrix}
    p_{k+1}\\v_{k+1}
\end{bmatrix} = 
\begin{bmatrix}
1 & T_s
\\
0 & 1
\end{bmatrix} 
\begin{bmatrix}
    p_{k} \\v_{k}
\end{bmatrix}
-
\begin{bmatrix}
    0 \\ T_s F(v_k)
\end{bmatrix}
+
\begin{bmatrix}
    0 \\ T_s
\end{bmatrix}
u_k,
\end{align*}
with the a-priori unknown friction term $F(v_k)$ given by 
\begin{subequations} 
    \label{eq:distribution_results}
\begin{align}
    F(v)\sim \mathcal{N}(h(v), g(v)),\label{eq:FrictionTerm}
\end{align}
with
\begin{align}
    h(v) &= v + \sin(3v),
    \\
    g(v) &= 0.8 \max(0,v) + 0.2.
\end{align} 
\end{subequations} 
Hence, the dynamics model studied in this section exhibits an uncertain friction term \eqref{eq:FrictionTerm}, which we learn using the proposed algorithm.
The GP model is subsequently used within a stochastic controller as state constraints, which are enforced as chance constraints, 
\begin{align}
\label{eq:chance_constraints}
    &{\rm Pr}(v_k \in \mathrm{X})\geq p_x ,
\end{align} 
with $\mathrm{X}=[-1,1]$m/s and confidence interval $p_x=95.44\%$, i.e., $\pm 2\sigma$ for a normal distribution.
The results in the following use GP models learned using $N=100$ data points, which are randomly generated according to $\dataIndex{v}{\dataCounter}\sim \mathcal{U}(-1,1)$ for each Monte-Carlo trial.
The optimal controller optimizes the cost function $\sum_{k=0}^{N} (10 (p_k-p_{\rm ref})^2 + u_k^2)$, where $p_{\rm ref}=1$ for 0--2s and $p_{\rm ref}=0$ after 2s.

Fig.~\ref{fig:stochastic_control} shows statistics of 500 Monte-Carlo trials of the stochastic controller using three different implementations of GP models.
It shows the stochastic controller using the proposed GP model and two GPs with constant priors as a comparison, one that uses a prior to remain cautious and one that uses a prior to be aggressive.
The aggressive GP implementation uses a prior $\sigma_0^2 = 0.2$, which is the smallest noise in the data~\eqref{eq:distribution_results}.
The cautious GP implementation uses $\sigma_0^2 = 1$, which is the highest noise level inside the chance constraints in~\eqref{eq:chance_constraints}.
Fig.~\ref{fig:stochastic_control} illustrates the median of the position and the velocity, as well as the $2\sigma$-bound of the velocity spread.
It shows that the proposed GP implementation is cautious between~0--1s, which is needed to satisfy the chance constraints in~\eqref{eq:chance_constraints}.
During this time interval, this behavior resembles the cautious GP very closely (see middle plot).
Next, Fig.~\ref{fig:stochastic_control} shows that the proposed GP implementation is aggressive during the time interval~2--3s, which is caused by the low noise characteristics for negative velocities.
During this time interval, this behavior resembles the aggressive GP very closely (see bottom plot).
On the other hand, the aggressive GP is too aggressive during~0--1s and the cautious GP is not aggressive enough during~2--3s.
Overall, the proposed GP is cautious when its needed to satisfy~\eqref{eq:chance_constraints}, but is aggressive whenever possible.

Table~\ref{tb:cost} summarizes the closed-loop cost and percentages of constraint violations of the three GP implementations. 
As expected, the aggressive GP implementation violates the chance constraints in~\eqref{eq:chance_constraints} with an average 22\% of constraints violations during intervals~0--1s and~2--3s.
Next, the cautious GP implementation is overly conservative, which can be seen by its higher closed-loop operating cost.
The proposed GP implementation's constraint violations matches the $2\sigma$-bound according to the requirements in the chance constraints in~\eqref{eq:chance_constraints}.
Consequently, its operating cost are lower than the cautious GP implementation.

   \begin{figure}[t]
    \centering   
\begin{tabular}{lll}
\begin{tikzpicture}[every plot/.append style={very thick}]
\begin{axis}[
    xmin=0, xmax=4,
    ymin=-0.1, ymax=1.49, 
   ylabel={Position in [m]}, 
    legend style={at={(0.5,1)},anchor=north}, 
    ymajorgrids=true,
    legend columns=3,
    grid style=dashed,
    width=8.3cm, 
    height=4cm,
	] 
]	  
\addplot[color=mst2!20] table[x index=0,y index=1] {data/toy_cautious.txt}; 
\addplot[color=mst3!20,] table[x index=0,y index=1] {data/toy_aggressive.txt}; 
\addplot[color=black,] table[x index=0,y index=1] {data/toy_data_proposed.txt}; 
\addplot[color=mst2,dashed] table[x index=0,y index=1] {data/toy_cautious.txt}; 
\addplot[color=mst3,dashdotted] table[x index=0,y index=1] {data/toy_aggressive.txt}; 

\legend{,,Proposed, Cautious, Aggressive};     
 
\end{axis}
\end{tikzpicture}   
    \\
      \begin{tikzpicture}[every plot/.append style={ thick}]
\begin{axis}[
    xmin=0, xmax=4,
    ymin=-1.25, ymax=1.25, 
   ylabel={Velocity in [m/s]}, 
    legend style={at={(1,1)},anchor=north east}, 
    ymajorgrids=true,
    legend columns=1,
    transpose legend,
    grid style=dashed,
    width=8.3cm, 
    height=4cm,
	] 
]

\addplot[mst2,fill=mst2!50,opacity=.7] table[x index=0,y index=1] {data/toy_v_std_cautious.txt};
\addplot[black!95,fill=gray!50,opacity=.7] table[x index=0,y index=1] {data/toy_v_std_proposed.txt};
\addplot[color=mst2,dashed] table[x index=0,y index=1] {data/toy_v_cautious.txt}; 
\addplot[color=black,] table[x index=0,y index=1] {data/toy_v_proposed.txt}; 

\legend{,,Cautious,Proposed};    
    
\end{axis}
\end{tikzpicture} 
    \\
      \begin{tikzpicture}[every plot/.append style={ thick}]
\begin{axis}[
    xmin=0, xmax=4,
    ymin=-1.25, ymax=1.25, 
   ylabel={Velocity in [m/s]}, 
   xlabel={Time in [s]}, 
    legend style={at={(1,1)},anchor=north east}, 
    ymajorgrids=true,
    legend columns=1,
    transpose legend,
    grid style=dashed,
    width=8.3cm, 
    height=4cm,
	] 
]	 
\addplot[mst3,fill=mst3!50,opacity=.7] table[x index=0,y index=1] {data/toy_v_std_agg.txt};
\addplot[black!95,fill=gray!50,opacity=.7] table[x index=0,y index=1] {data/toy_v_std_proposed.txt};
\addplot[color=mst3,dashdotted] table[x index=0,y index=1] {data/toy_v_agg.txt}; 
\addplot[color=black,] table[x index=0,y index=1] {data/toy_v_proposed.txt}; 

\legend{,,Aggressive,Proposed};    
    
\end{axis}
\end{tikzpicture} 
\end{tabular} 
    \caption{Qualitative performance results of stochastic control using different GP models, i.e., three stochastic controllers.
    Top: Position achieved by the three different stochastic controllers.
    Middle: Velocity profile of stochastic controller using proposed GP with state-dependent noise inference (in black) and stochastic controller using GP with cautious choice of constant prior (in yellow).
    Bottom: Velocity profile of stochastic controller using proposed GP with state-dependent noise inference (in black) and stochastic controller using GP with aggressive choice of constant prior (in purple). 
    The figure shows the median of 500 Monte-Carlo runs as well as the $2\sigma$ bound.}
    \label{fig:stochastic_control}
\end{figure}
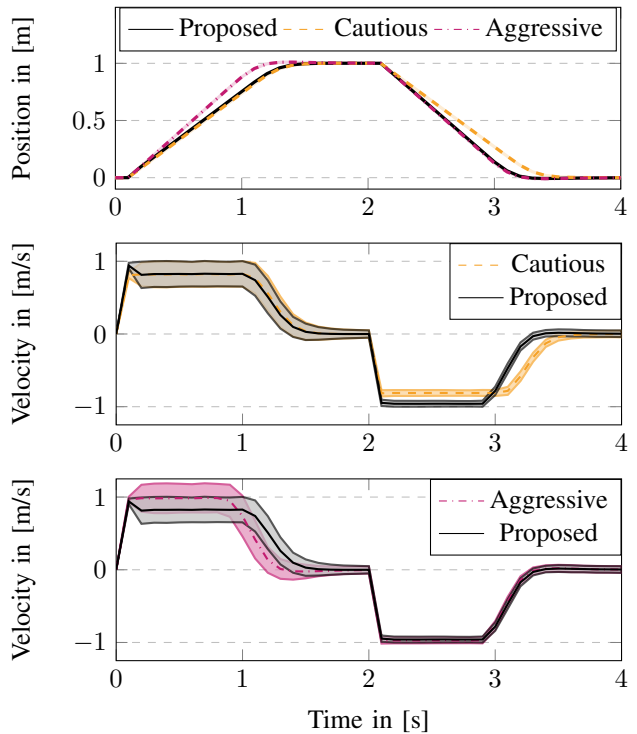

\renewcommand{\arraystretch}{0.825}
\begin{table}[t]
\caption{Cost and Constraint Violations}
\vspace{-0.5cm}
\label{tb:cost}
\begin{center}
\begin{tabular}{l|ccc|cccc}
\toprule %
 & & \hspace{-1cm} Closed-loop Cost \hspace{-1cm} & & Constraint Violations\\
Approach & Min. & Median & Max.  & during 0--1s and 2--3s\\
\midrule 
Proposed & 82.2 & 86.8 & 92.1 & 2\% \\
Aggressive & 77.1 & 81.3 & 86.6 & 22\% \\   
Cautious & 89.0 & 94.1 & 100.2 & 1\%\\   
\bottomrule
\end{tabular}
\end{center}
\end{table}

\section{Conclusions}
This paper considered a stochastic control formulation, in which the residual model uncertainty of the dynamical system consists of a nonlinear function and state-dependent noise.
It presented an iterative algorithm to learn the unknown nonlinear function along with the state-dependent noise distribution using GP regression.
The iterative algorithm used a posterior-GP to approximate the residual model uncertainty and a prior-GP to account for state-dependent noise. 
The iterative algorithm was shown to converge under simplifying assumptions. 
Simulation results showed that the proposed GP formulation enables faster convergence of the GP as the algorithm leverages the trustworthiness of the data by means of the prior,
Further, simulation results showed the advantages of using the proposed GP formulation within a stochastic control framework.

\addtolength{\textheight}{-7.5cm}  

\section*{Appendix: Proofs}

\subsubsection*{Proof of Lemma~\ref{lma:bounded}}
The update equations at iteration~$j$ used in Algorithm~\ref{alg:alg} are given by~\eqref{eq:iterations}.
Approximating~\eqref{eq:iterations} with a small length scale, $\funK{\X}{\X} \rightarrow \I$, yields
\begin{align*}
    &\meanGPmeanVec_{j+1}(\X)
    \approx \left(\left(1 \!+\!\sigma_0^2\right)\I + \priorGPmeanDiag_{j}(\X)\right)^{-1} \Y
\end{align*}
and
\begin{align*}
    &
    \priorGPmeanVec_{j+1}(\X)
    \approx \frac{1}{1+\sigma_0^2} 
    \left(\left(\Y - \meanGPmeanVec_{j+1}(\X) \right)^{\circ 2} - \sigma_0^2 \mathbf{1}\right)
    \\ 
    & 
    = \frac{1}{1+\sigma_0^2} 
    \left(\left(\Y-\left(\left(1+\sigma_0^2\right)\I + \priorGPmeanDiag_j(\X) \right)^{-1} \Y \right)^{\circ 2} - \sigma_0^2 \1\right)
    \\
    &
    = \frac{1}{1+\sigma_0^2} 
    \left(\left(\I-\left(\left(1+\sigma_0^2\right)\I + \priorGPmeanDiag_j(\X) \right)^{-1}\right)^2 \Y^{\circ 2} - \sigma_0^2 \1\right)
\end{align*}
Exploiting the diagonal structure of the matrix inverse, an analytical expression can be formulated,
\begin{align} 
\label{eq:analytical_z_new}
    \priorGPmeanVec_{j+1}(\X) 
    =
    \frac{1}{1+\sigma_0^2} 
    \left(\left( \M(\priorGPmeanVec_j(\X)) \right)^2 \Y^{\circ 2} - \sigma_0^2 \1\right)
\end{align}
with the diagonal matrix $\M(\priorGPmeanVec_j(\X))$, whose element on the $\elementCounterB$-th row and column is given by
\begin{align*}
    \getElementOfVec{\M(\priorGPmeanVec_j(\X))}{\elementCounterB\elementCounterB}
    =
    \frac{
    \sigma_0^2 \!+ \getElementOfVec{\priorGPmeanVec_j(\X)}{\elementCounterB}
    }{
    1 \!+\! \sigma_0^2 \!+\! \getElementOfVec{\priorGPmeanVec_j(\X)}{\elementCounterB}
    }.
\end{align*} 
 
We proceed using induction-based arguments.
Suppose $- \frac{\sigma_0^2}{1+\sigma_0^2}\leq 
\getElementOfVec{\priorGPmeanVec_j(\X)}{\elementCounterB}$, then, it is easy to verify that bounds can be established with
\begin{align}
\label{eq:bounds_on_M_new}
    0
    \leq 
    \getElementOfVec{\M(\priorGPmeanVec_j(\X))}{\elementCounterB\elementCounterB}
    =
    \frac{\sigma_0^2 \! + \getElementOfVec{\priorGPmeanVec_j(\X)}{\elementCounterB} 
    }{
    1 \!+\! \sigma_0^2 \!+\! \getElementOfVec{\priorGPmeanVec_j(\X)}{\elementCounterB}
    }
    \leq 1.
\end{align}
Using the lower bound and the upper bound in~\eqref{eq:bounds_on_M_new}, the expression in~\eqref{eq:analytical_z_new} can be bounded with
\begin{align*}
    - \frac{\sigma_0^2}{1+\sigma_0^2}\1 & \leq \priorGPmeanVec_{j+1}(\X) 
    \leq 
    \frac{1}{1+\sigma_0^2} (  \Y^{\circ 2} - \sigma_0^2 \1 ).
\end{align*}
Finally, the Algorithm uses the initialization $\priorGPmeanVec_{0}(\X)=\0$, which completes the induction.\hfill $\blacksquare$

\subsubsection*{Proof of Theorem~\ref{thm:contracting}}
For the converged optimizer of Algorithm~\ref{alg:alg},
\begin{align}
    \priorGPmeanVec_\star (\X)
    =
    \frac{1}{1+\sigma_0^2} 
    \left(\left( \M(\priorGPmeanVec_\star(\X)) \right)^2 Y^{\circ 2} - \sigma_0^2 \1\right)
\end{align}
with the diagonal matrix $\M(\priorGPmeanVec_\star(X))$ whose element on the $i$th row and column is given by
\begin{align*}
\getElementOfVec{\M (\priorGPmeanVec_\star(\X))}{\elementCounterB\elementCounterB}
    =
    \frac{
    \sigma_0^2 \!+  \getElementOfVec{\priorGPmeanVec_\star(\X)}{\elementCounterB}
    }{
    1 \!+\! \sigma_0^2 \!+\! \getElementOfVec{\priorGPmeanVec_\star(\X)}{\elementCounterB}
    }.
\end{align*} 

Next, consider the difference of the optimizer update to the converged optimizer, $\priorGPmeanVec_{j+1}(\X) - \priorGPmeanVec_\star(\X)$.
To ease exposition, let $\priorGPmean_{j}^\elementCounterB :=\getElementOfVec{\priorGPmeanVec_{j}(\X)}{\elementCounterB}$ and $\priorGPmean_\star^\elementCounterB := \getElementOfVec{\priorGPmeanVec_\star(\X)}{\elementCounterB}$.
For the $\elementCounterB$-th element, and using some algebraic reformulations
\begin{align*}
& 
\getElementOfVec{\priorGPmeanVec_{j+1}(X) - \priorGPmeanVec_\star(X)}{\elementCounterB} 
    = \priorGPmean_{j+1}^\elementCounterB - \priorGPmean^\elementCounterB_\star %
    \\
    &=
    \frac{1}{1\!+\!\sigma_0^2} 
    \left(
    \hspace{-0.1cm}
    \left(
    \frac{\sigma_0^2 \!+\! \priorGPmean_{j}^\elementCounterB}{1\!+\!\sigma_0^2 \!+\! \priorGPmean_{j}^\elementCounterB}
    \right)^2 
    \hspace{-0.1cm}
    (\dataIndex{y}{\elementCounterB})^2
     -  
    \left(
    \frac{\sigma_0^2 \!+\! \priorGPmean_\star^\elementCounterB}{1\!+\!\sigma_0^2 \!+\! \priorGPmean_\star^\elementCounterB}
    \right)^2 
    \hspace{-0.1cm}
    (\dataIndex{y}{\elementCounterB})^2
    \hspace{-0.1cm}
    \right) 
    \\
    &=
    \frac{(\dataIndex{y}{\elementCounterB})^2}{1\!+\!\sigma_0^2} 
    \frac{
    (\sigma_0^2 \!+\! \priorGPmean^\elementCounterB_j)^2 (1\!+\!\sigma_0^2 \!+\! \priorGPmean^\elementCounterB_\star)^2 
    \!-\! 
    (\sigma_0^2 
    \!+\! 
    \priorGPmean^\elementCounterB_\star)^2
    (1\!+\!\sigma_0^2 \!+\! \priorGPmean^\elementCounterB_j)^2
    }{
    (1\!+\!\sigma_0^2 \!+\! \priorGPmean^\elementCounterB_j)^2 (1\!+\!\sigma_0^2 \!+\! \priorGPmean^\elementCounterB_\star)^2}  
    \\
    &=
    \underbrace{
    \frac{(\dataIndex{y}{\elementCounterB})^2}{1\!+\!\sigma_0^2} 
    \frac{
    \priorGPmean^\elementCounterB_j \!+\! \priorGPmean^\elementCounterB_\star \!+\!2( \sigma_0^2 \!+\! \sigma_0^2 \priorGPmean^\elementCounterB_j \!+\! \sigma_0^2 \priorGPmean^\elementCounterB_\star \!+\! \priorGPmean^\elementCounterB_j \priorGPmean^\elementCounterB_\star \!+\! \sigma_0^4)
    }{
    (1\!+\!\sigma_0^2 \!+\! \priorGPmean^\elementCounterB_j)^2 (1\!+\!\sigma_0^2 \!+\! \priorGPmean^\elementCounterB_\star)^2
    }}_{=:a^\elementCounterB_j  } %
     (\priorGPmean^\elementCounterB_j \!-\! \priorGPmean^\elementCounterB_\star)
\end{align*} 
Note that the reformulation in the last line can be obtained by polynomial long division. %

The proof proceeds by establishing bounds for $a^\elementCounterB_j$.
Using the design choice $\sigma_0^2 \geq \max_\elementCounterB\ (\dataIndex{y}{\elementCounterB})^2$ and the bounds established in Lemma~\ref{lma:bounded},
\begin{align}
    \label{eq:bound_with_sigma}
    - \frac{\sigma_0^2}{1+\sigma_0^2}\1 & \leq \priorGPmeanVec_{j+1}(\X)   \leq \0.
\end{align} 

\textit{Upper bound on $a_j^\elementCounterB$:}
In order to derive the upper bound on $a_j^\elementCounterB$, we use~\eqref{eq:bound_with_sigma} to find the worst-case nominator with $\priorGPmean_j^\elementCounterB\leq 0$, $\priorGPmean_\star^\elementCounterB \leq 0$, $\priorGPmean_j^\elementCounterB \priorGPmean_\star^\elementCounterB \leq (\frac{\sigma_0^2}{1+\sigma_0^2})^2$, and $(\dataIndex{y}{\elementCounterB})^2 \leq \sigma_0^2$; and the worst-case denominator with $-\frac{\sigma_0^2}{1+\sigma_0^2}\leq \priorGPmean^\elementCounterB_j$ and $-\frac{\sigma_0^2}{1+\sigma_0^2}\leq \priorGPmean^\elementCounterB_\star$.
Hence, %
\begin{align*}
    a_j^\elementCounterB
    &\leq
    \frac{\sigma_0^2}{1+\sigma_0^2} 
    \frac{
    2 \sigma_0^2 + 2 \left(\frac{\sigma_0^2}{1+\sigma_0^2}\right)^2 + 2\sigma_0^4
    }{\left(1 +\sigma_0^2 - \frac{\sigma_0^2}{1+\sigma_0^2} \right)^2 \left(1+\sigma_0^2 - \frac{\sigma_0^2}{1+\sigma_0^2} \right)^2} 
    \\
    &=
    \underbrace{
    \frac{\sigma_0^2}{1+\sigma_0^2} 
    }_{< 1}
    \underbrace{
    \frac{
    2\sigma_0^2 (\sigma_0^2+1)^2 (\sigma_0^6 + 3\sigma_0^4 + 4\sigma_0^2 +  1)
    }{ \left(\sigma_0^4 + \sigma_0^2 + 1\right)^4}
    }_{< 1} <1.
\end{align*} 
Note that the second fraction being less than one is not easy to see, but can be easily established using any numerical software.

\textit{Lower bound on $a_j^\elementCounterB$:}
In order to derive the lower bound on $a_j^\elementCounterB$, we use~\eqref{eq:bound_with_sigma} with $\priorGPmean_j^\elementCounterB \geq -\frac{\sigma_0^2}{1+\sigma_0^2}$, $\priorGPmean^\elementCounterB_\star \geq -\frac{\sigma_0^2}{1+\sigma_0^2}$, and $\priorGPmean_j^\elementCounterB \priorGPmean^\elementCounterB_\star \geq - (\frac{\sigma_0^2}{1+\sigma_0^2})^2$ for the nominator; and $-\frac{\sigma_0^2}{1+\sigma_0^2}\leq \priorGPmean^\elementCounterB_j$ and $-\frac{\sigma_0^2}{1+\sigma_0^2}\leq \priorGPmean^\elementCounterB_\star$ for the denominator, i.e.,
\begin{align*} 
a_j^\elementCounterB
& \geq 
\frac{\sigma_0^2}{1+\sigma_0^2} 
    \frac{ 
    - 2\frac{\sigma_0^2}{1+\sigma_0^2} 
    \!+\!
    2\left(
    \sigma^2_0
    \!-\!
    2
    \sigma^2_0
    \frac{\sigma_0^2}{1+\sigma_0^2} 
    \!-\! \left(\frac{\sigma_0^2}{1+\sigma_0^2}\right)^2
    \!+\! \sigma_0^4
    \right)
    }{\left(1 +\sigma_0^2 - \frac{\sigma_0^2}{1+\sigma_0^2} \right)^2 \left(1+\sigma_0^2 - \frac{\sigma_0^2}{1+\sigma_0^2} \right)^2} 
    \\
    & = \underbrace{\frac{\sigma_0^2}{1+\sigma_0^2} }_{<1}
    \underbrace{
    \frac{2\sigma_0^4(\sigma_0^2+1)^2 (\sigma_0^4+\sigma_0^2 - 1)}{\left(\sigma_0^4 + \sigma_0^2 + 1\right)^4} }_{>-1}
    > -1.
\end{align*}
Hence, overall,
\begin{align*}
    \priorGPmeanVec_{j+1}(\X) - \priorGPmeanVec_\star(\X)  
    =
    \boldsymbol A_j  
     (\priorGPmeanVec_j(\X) - \priorGPmeanVec_\star(\X))
\end{align*}
with the diagonal matrix $\boldsymbol A_j$, whose diagonal elements $-1<a^\elementCounterB_j<1$, and thus, $\priorGPmeanVec_{j}(\X) \rightarrow \priorGPmeanVec_\star(\X)$ for $j\rightarrow\infty$. \hfill $\blacksquare$

\bibliography{bib.bib}
\bibliographystyle{ieeetr}

\end{document}